\documentclass[a4paper,UKenglish]{lipics-v2016}

\usepackage{microtype}%
\usepackage{pdfpages}
 
\title{Faster Algorithms for the Maximum Common Subtree Isomorphism Problem%
\footnote{This work was supported by the German Research Foundation (DFG), priority programme ``Algorithms for Big Data'' (SPP 1736).}}

\author[1]{Andre Droschinsky}
\author[1]{Nils M. Kriege}
\author[1]{Petra Mutzel}
\affil[1]{Dept.\ of Computer Science,
Technische Universit{\"a}t Dortmund, Germany\\
  \texttt{\{andre.droschinsky,nils.kriege,petra.mutzel\}@tu-dortmund.de}}

\authorrunning{A. Droschinsky and N.\,M. Kriege and P. Mutzel}

\Copyright{Andre Droschinsky and Nils M. Kriege and Petra Mutzel}



\EventYear{2016}
\EventLogo{}
\ArticleNo{1}

\usepackage{xspace}
\usepackage{tikz}

\usepackage[english,linesnumbered,vlined,ruled,nokwfunc]{algorithm2e}
\DontPrintSemicolon
\SetKwComment{note}{$\triangleright$ }{}
\SetFuncSty{textsc}
\SetCommentSty{textit}
\SetDataSty{texttt}
\SetNlSty{text}{\color{black!100!white}}{}

\SetKwInOut{Input}{Input}
\SetKwInOut{Data}{Data}
\SetKwInOut{Output}{Output}
\ResetInOut{Output} 

\SetKw{KwAnd}{and} %

\newcommand{\wMatch}[0]{\ensuremath{w}\xspace}
\newcommand{\WMatch}[0]{\ensuremath{W}\xspace}
\def\MWM{MWM\xspace}
\def\MWPM{MWPM\xspace}
\def\MWPMs{MWPMs\xspace}
\DeclareMathOperator{\dom}{dom}
\newcommand{\wIso}[0]{\ensuremath{\omega}\xspace}
\newcommand{\WIso}[0]{\ensuremath{\mathcal{W}}\xspace}
\def\MCSI{MCSI\xspace}
\def\MWMs{MWMs\xspace}
\def\MCSR{\textsc{MCS}_\text{root}}
\def\MCSIs{MCSIs\xspace}

\theoremstyle{plain}
\newtheorem{proposition}[theorem]{Proposition}

\begin{document}

\maketitle

\begin{abstract}
The maximum common subtree isomorphism problem asks for the largest possible 
isomorphism between subtrees of two given input trees.
This problem is a natural restriction of the maximum common subgraph problem, which is ${\sf NP}$-hard in general graphs.
Confining to trees renders polynomial time algorithms possible and is of 
fundamental importance for approaches on more general graph classes.
Various variants of this problem in trees have been intensively studied.
We consider the general case, where trees are neither rooted nor ordered and the 
isomorphism is maximum w.r.t.\ a weight function on the mapped vertices and edges.
For trees of order $n$ and maximum degree $\Delta$ our algorithm achieves a running time of $\mathcal{O}(n^2\Delta)$ by exploiting the structure of the 
matching instances arising as subproblems.
Thus our algorithm outperforms the best previously known approaches.
No faster algorithm is possible for trees of bounded degree and
for trees of unbounded degree we show that a further reduction of the running time 
would directly improve the best known approach to the assignment problem. 
Combining a polynomial-delay algorithm for the enumeration 
of all maximum common subtree isomorphisms with central ideas of our new algorithm
leads to an improvement of its running time from $\mathcal{O}(n^6+Tn^2)$ to $\mathcal{O}(n^3+Tn\Delta)$, where $n$ is the order of the larger tree, $T$ is the number of different solutions, and $\Delta$ is the minimum of the maximum degrees of the input trees.
Our theoretical results are supplemented by an experimental evaluation on synthetic
and real-world instances.
 \end{abstract}

\section{Introduction}
The maximum common subgraph isomorphism problem (MCS) asks for an isomorphism 
between induced subgraphs of two given graphs that is of maximum weight
w.r.t.\ a weight function on the mapped vertices and edges.
The problem is of fundamental importance in applications like pattern 
recognition~\cite{Conte2004} or bio- and cheminformatics~\cite{Ehrlich2011,Schietgat2013}.
MCS naturally generalizes the subgraph isomorphism problem 
(SI), where the task is to decide if one graph is isomorphic to a subgraph of 
another graph. Both problems are known to be ${\sf NP}$-hard for general graphs.

It is not astonishing that these problems have been extensively studied for restricted graph classes.
Polynomial time algorithms for SI and MCS in trees have been pioneered by Edmonds 
and Matula in the 1960s. They rely on solving a series of maximum bipartite 
matching instances, see~\cite{Matula1978}.
These early results focused on the polynomial time complexity of the problem; 
since then considerable progress has been achieved in improving the running time 
of SI algorithms (also see \cite{Abboud2016} and references therein): 
Reyner \cite{Reyner1977,Verma1989} and Matula \cite{Matula1978} both showed a 
running time of $\mathcal{O}(n^{2.5})$ for rooted trees. Chung~\cite{Chung1987} 
later obtained the same bound for unrooted trees.
Further improvements were made by Shamir and Tsur~\cite{Shamir1999} who obtained 
time $\mathcal{O}(n^{2.5}/\log n)$ and $\mathcal{O}(n^\omega)$, where $\omega$ is 
the exponent of matrix multiplication.

MCS on trees seems to be harder. For two rooted trees of size $n$, 
it is known that the problem can be solved, roughly speaking, in the same time 
as the associated maximum weight matching problem in a bipartite graph on $n$ 
vertices:
Gupta and Nishimura~\cite{Gupta1998} presented an $\mathcal{O}(n^{2.5} \log n)$ 
algorithm for MCS in rooted trees by assuming weights to be in $\mathcal O(n)$, which 
allows to employ a scaling approach to the matching problem~\cite{GT1989}.
The running time can be improved to $\mathcal{O}(\sqrt{\Delta}n^2\log{\frac{2n}{\Delta}})$, 
where $\Delta$ denotes the maximum degree~\cite{Kao2001}. 
Allowing a real weight function to determine the similarity of mapped vertices
gives rise to bipartite matching instances with unrestricted weights.
Solving these with the Hungarian method leads to cubic running time, see 
e.g.~\cite{Valiente2002}. Since the the size of the matching instances is 
bounded by the maximum degree $\Delta$, the result can be improved to 
$\mathcal{O}(n^2\Delta)$ time~\cite{Torsello2005}.
Various related concepts for the comparison of rooted trees, either ordered or 
unordered, have been proposed and were studied in detail, see \cite{Valiente2002}
and references therein, where the \emph{tree edit distance} is a prominent 
example~\cite{Akutsu2013a,Demaine2009}.

In this article we consider the problem of finding a common subtree isomorphism 
in unrooted, unordered trees that is maximum w.r.t.\ a weight function on the 
mapped vertices and edges.
This problem is directly relevant in various applications, where real-world objects 
like molecules or shapes are represented by (attributed) trees~\cite{FTrees,Torsello2005}.
Moreover, it forms the basis for several recent approaches to solve MCS in more 
general graph classes, see~\cite{Akutsu2013,Kriege2014a,Kriege2014b,Schietgat2013}.
Methods directly based on algorithms for rooted trees result in time 
$\mathcal{O}(n^4\Delta)$ by considering all pairs of possible roots. 
An improvement to $\mathcal{O}(n^3\Delta)$ has been reported in \cite{Torsello2005},
which is limited to non-negative weight functions.
Schietgat, Ramon and Bruynooghe~\cite{Schietgat2013} suggested an approach for 
MCS in outerplanar graphs, which solves the considered problem when applied to 
trees. The approach is stated to have a running time of $\mathcal{O}(n^{2.5})$,
but in fact leads to a running time of $\Omega(n^4)$ in the worst case.%
\footnote{The analysis of the algorithm appears to be flawed. An Erratum to 
\cite{Schietgat2013} has been submitted to the \emph{Annals of Mathematics and 
Artificial Intelligence}, see Appendix.
Our experimental study of their implementation actually suggests a time bound of $\Omega(n^5)$.}

\subparagraph{Our contribution.}
We show that for arbitrary weights a maximum common subtree isomorphism between two trees $G$ and $H$ of order $n$ with $\Delta(G)\leq \Delta(H)$ can be computed in time $\mathcal{O}(n^2 (\Delta(G)+\log \Delta(H))$.
We obtain the improvement
by (i) considering only a specific subset of subproblems that we show to be 
sufficient to guarantee an optimal solution; (ii) exploiting the close relation 
between the emerging matching instances.
We show that for general trees any further improvement of this time bound would allow to solve the
assignment problem in $o(n^3)$, and hence improve over the best known approach
to this famous problem for more than 30 years. 
For trees of bounded degree the running time bound of $\mathcal O(n^2)$ is tight.
We apply our new techniques to the problem of enumerating all maximum common subtree isomorphisms, thus improving the state-of-the-art running times. 
Finally, we present an experimental
evaluation on synthetic and real-world instances showing that our new algorithm is faster than existing approaches.

\section{Preliminaries}\label{34_sec:Preliminaries}
In this paper, $G=(V,E)$ is a \emph{simple undirected graph}.
We call $v\in V$ a \emph{vertex} and  $uv=vu\in E$ an \emph{edge} of $G$.
For a graph $G=(V,E)$ we define $V(G):=V_G:=V$, $E(G):=E_G:=E$ and $|G|:=|V(G)|$.
For a subset of vertices $V'\subseteq V$ the graph $G[V']:=(V',E'),\ E':=\{uv\in E\mid u,v\in V'\}$, is called \emph{induced} subgraph.
A connected graph with a unique path between any two vertices is a \emph{tree}.
A tree $G$ with an explicit root vertex $r\in V_G$ is called \emph{rooted} tree, denoted by $G^r$.
In a rooted tree $G^r$ we denote the children of a vertex $v$ by $C(v)$
and its parent by $p(v)$, where $p(r) = r$.
The \emph{depth} ${\rm depth}(v)$ of a vertex $v$ is the number of edges on the path from 
$v$ to $r$.
The \emph{neighbors} of a vertex $v$ are defined as $N(v):=\{u\in V_G\mid vu\in E_G\}$.
The \emph{degree} of a vertex $v\in V_G$ is $\delta(v):=|N(v)|$, the \emph{degree} $\Delta(G)$ of a graph $G$ is the maximum degree of its vertices.

For a graph $G=(V,E)$ a \emph{matching} $M\subseteq E$ is a set of edges, such that no two edges share a vertex.
A matching $M$ of $G$ is said to be \emph{perfect}, if $2|M|=|V|$.
A \emph{weighted graph} is a graph endowed with a function $\wMatch:E\to\mathbb R$.
The weight of a matching $M$ in a weighted graph is $\WMatch(M):=\sum_{e\in M} \wMatch(e)$.
We call a matching $M$ of a weighted bipartite graph $G$ a \emph{maximum weight matching} (\MWM) if there is no other matching $M'$ of $G$ with $\WMatch(M')>\WMatch(M)$.
The \emph{assignment problem} asks for a matching with maximum weight among all
perfect matchings and we refer to a solution by \MWPM.

An \emph{isomorphism} between two graphs $G$ and $H$ is a bijective function $\phi:V_G\to V_H$ such that $uv\in E_G\Leftrightarrow\phi(u)\phi(v)\in E_H$;
if such an isomorphism exists, $G$ and $H$ are said to be \emph{isomorphic}.
We call a graph $G$ \emph{subgraph isomorphic} to a graph $H$, if there is an induced subgraph $H'\subseteq H$ isomorphic to $G$.
In this case, we write  $G\preceq_\phi H$, where $\phi:V_G\to V_{H'}$ is an isomorphism between $G$ and $H'$.
A \emph{common subgraph} of $G$ and $H$ is a graph $I$, such that $I\preceq_\phi G$ and $I\preceq_{\phi'} H$.
The isomorphism $\varphi:=\phi'\circ\phi^{-1}$ is called \emph{common subgraph isomorphism} (CSI).
For a function $f:X\to Y$ let $\dom(f):=X$ be the domain of $f$.
If there is no other CSI $\varphi'$ with $|\dom(\varphi')|>|\dom(\varphi)|$, we 
call $\varphi$ \emph{maximum} common subgraph isomorphism (MCSI).

We generalize the above definitions to a pair of graphs $G,H$ under a commutative weight-function $\wIso:(V_G\times V_H)\cup (E_G\times E_H)\to \mathbb{R}\cup \{-\infty\}$.
The weight $\WIso(\phi)$ of an isomorphism $\phi$ between $G$ and $H$ under $\wIso$ is the sum of the weights $\wIso(v,\phi(v))$ and $\wIso(vu,\phi(v)\phi(u))$ of all vertices and edges mapped by $\phi$.
A maximum common subgraph isomorphism $\phi$ under a weight function is one of maximum weight $\WIso(\phi)$ instead of maximum size $|\dom(\phi)|$.
Note, this is less restrictive than a common approach for isomorphisms on labeled graphs, where the labels must match.
By defining $\wIso$ such that mapped vertices and mapped edges add 1 and 0, respectively, to the weight, we obtain an isomorphism of maximum size.
Therefore, in the following we consider graphs under a weight function unless stated otherwise and refer to the corresponding solution as \MCSI. 
For convenience we replace the word \emph{graph} by \emph{tree} in the above
definitions when appropriate. 
The \emph{maximum common subtree isomorphism problem} (MCST) is to determine the
weight of an \MCSI, where the input graphs and the common subgraph are trees.
We further define $[1..k]:=\{1,\ldots,k\}$ for $k\in \mathbb N$.

\section{Problem Decomposition and Fundamental Algorithms}\label{34_sec:FDPF}
We introduce the basic techniques for solving MCST
following the ideas of Edmonds and Matula~\cite{Matula1978}. 
The approach requires to compute \MWMs in bipartite graphs as a subroutine. We
discuss the occurring matching instances in detail in Section~\ref{34_sec:MWM}.

By fixing the roots of both trees we can develop an algorithm solving MCST on this restricted setting. It is easy to generalize this solution by considering all possible pairs of roots. We then show that it is sufficient to fix the root of one tree while still obtaining a maximum solution.

\subparagraph{Rooted trees.}

\tikzstyle{vertex_34}=[circle, draw=black, fill=gray!80!white,semithick,scale=.8]
\tikzstyle{rvertex_34}=[circle, draw=black, fill=black,semithick,scale=.8]
\tikzstyle{svertex_34}=[circle, draw=gray!80!white, fill=gray!20!white,semithick,scale=.8]
\tikzstyle{edge_34}=[draw=black, thick,-]
\tikzstyle{sedge_34}=[draw=black!30!white, thick,-]
\begin{figure}
   \centering
   \begin{subfigure}[b]{0.25\textwidth}
\begin{tikzpicture}[scale=0.6, auto]

    \node[svertex_34, label ={below:}] (u1) at (-3.5,-2) {};
    \node[svertex_34, label ={below:}] (u2) at (-2.5,-2) {};
    \node[svertex_34, label ={below:}] (u3) at (-1.5,-2) {};
    \node[svertex_34, label ={left:$r$}] (u4) at (-2.5,-1) {};
    \node[rvertex_34, label ={right:$u$}] (u5) at (-1,0) {};
    \node[vertex_34, label ={right:$c_1$}] (u6) at (-1,-1) {};
    \node[vertex_34, label ={right:$c_2$}] (u7) at (0.5,-1) {};
    \node[vertex_34, label ={right:}] (u8) at (0.5,-2) {};
    \node[vertex_34, label ={right:}] (u9) at (-0.3,-3) {};
    \node[vertex_34, label ={right:}] (u10) at (1.3,-3) {};

	\draw[sedge_34] (u1) -- (u4);
	\draw[sedge_34] (u2) -- (u4);
	\draw[sedge_34] (u3) -- (u4);
	\draw[sedge_34] (u4) -- (u5);
	\draw[edge_34] (u5) -- (u6);
	\draw[edge_34] (u5) -- (u7);
	\draw[edge_34] (u7) -- (u8);
	\draw[edge_34] (u8) -- (u9);
	\draw[edge_34] (u8) -- (u10);
\end{tikzpicture}
   \subcaption{Rooted subtree $G^r_u$}
   \label{34_example:mcs-tree:graph1}
   \end{subfigure}
   \begin{subfigure}[b]{0.24\textwidth}
\begin{tikzpicture}[scale=0.6, auto]

    \node[vertex_34, label ={below:$d_1$}] (u1) at (-0.5,-1) {};
    \node[vertex_34, label ={below:$d_2$}] (u2) at (0.5,-1) {};
    \node[svertex_34, label ={below:$s$}] (u3) at (1.5,-1) {};
    \node[rvertex_34, label ={right:$v$}] (u4) at (1,0) {};
    \node[vertex_34, label ={right:$d_3$}] (u5) at (2.5,-1) {};
    \node[vertex_34, label ={right:}] (u6) at (2.5,-2) {};
    \node[vertex_34, label ={right:}] (u7) at (1.5,-3) {};
    \node[vertex_34, label ={right:}] (u8) at (3.5,-3) {};

	\draw[edge_34] (u1) -- (u4);
	\draw[edge_34] (u2) -- (u4);
	\draw[sedge_34] (u3) -- (u4);
	\draw[edge_34] (u4) -- (u5);
	\draw[edge_34] (u5) -- (u6);
	\draw[edge_34] (u6) -- (u7);
	\draw[edge_34] (u6) -- (u8);
\end{tikzpicture}
\subcaption{Rooted subtree $H^s_v$}
   \label{34_example:mcs-tree:graph2}
   \end{subfigure}
   \begin{subfigure}[b]{0.24\textwidth}
\begin{tikzpicture}[scale=0.58,auto]

    \node[vertex_34, label ={left:$c_1$}] (u6) at (0,0.5) {};
    \node[vertex_34, label ={left:$c_2$}] (u7) at (0,-0.5) {};

    \node[vertex_34, label ={right:$d_1$}] (v1) at (3,1) {};
    \node[vertex_34, label ={right:$d_2$}] (v2) at (3,0) {};
    \node[vertex_34, label ={right:$d_3$}] (v5) at (3,-1) {};

	\draw[edge_34] (u6) -- node {} (v1); 
	\draw[edge_34] (u6) -- node {} (v2); 
	\draw[edge_34] (u6) -- (v5); 

	\draw[edge_34] (u7) -- (v1); 
	\draw[edge_34] (u7) -- (v2); 
	\draw[edge_34, ultra thick] (u7) -- node[yshift=-15,xshift=-8] {4} (v5); 
\end{tikzpicture}
\subcaption{Matching problem}
   \label{34_example:mcs-tree:mwbm}
\end{subfigure}
   \begin{subfigure}[b]{0.24\textwidth}
\begin{tikzpicture}[scale=0.58,auto]

    \node[vertex_34, label ={left:$c_1$}] (u6) at (0,0.5) {};
    \node[vertex_34, label ={left:$c_2$}] (u7) at (0,-0.5) {};

    \node[vertex_34, label ={right:$d_1$}] (v1) at (3,1) {};
    \node[vertex_34, label ={right:$d_2$}] (v2) at (3,0) {};
    \node[vertex_34, label ={right:$d_3$}] (v5) at (3,-1) {};

	\draw[edge_34,draw=cyan] (u6) -- node {\textcolor{cyan}{1}} (v1); 
	\draw[edge_34] (u6) -- node {} (v2); 
	\draw[edge_34] (u6) -- (v5); 

	\draw[edge_34] (u7) -- (v1); 
	\draw[edge_34] (u7) -- (v2); 
	\draw[edge_34, ultra thick,draw=cyan] (u7) -- node[yshift=-15,xshift=-8] {\textcolor{cyan}{4}} (v5); 
\end{tikzpicture}
   \subcaption{MWM}
   \label{34_example:mcs-tree:mwbmsol}
    \end{subfigure}
   \caption{Two rooted subtrees~(\subref{34_example:mcs-tree:graph1}) and
     (\subref{34_example:mcs-tree:graph2}), the associated weighted matching 
     instance~(\subref{34_example:mcs-tree:mwbm}), and an \MWM on that instance~(\subref{34_example:mcs-tree:mwbmsol}).
     Light gray vertices and edges are not part of the rooted subtrees, root 
     vertices are shown in solid black. The maximum weight matching is shown in blue.
     We assume a weight function $\wIso$ with
     $\wIso(u,v) = 1$ for all $(u,v) \in V_G \times V_H$ and $\wIso(e,f) = 0$
     for all $(e,f) \in E_G \times E_H$. The edges without label 
     in~(\subref{34_example:mcs-tree:mwbm}) have weight $1$.
   }
   \label{34_example:mcs-tree}
\end{figure}

We first consider the problem restricted to rooted trees under the assumption that
the roots of the two trees must be mapped to each other. 
For a rooted tree $G^r$ we define the \emph{rooted subtree} $G^r_u$ as the subtree 
induced by $u$ and all its descendants in $G^r$ that is rooted at $u$, cf. Figures~\ref{34_example:mcs-tree}\subref{34_example:mcs-tree:graph1} and~\subref{34_example:mcs-tree:graph2}.
Note that $G^r_r = G^r$ and that $G^r_u$ and $G^s_u$ both refer to the 
same subtree unless $s$ is contained in $G^r_u$. 
The key to solving MCST 
for two rooted trees $G^r$ and $H^s$ is the following recursive formulation:
\begin{equation}\label{34_eq:basic_recursion}
 \MCSR(G^r,H^s) = \wIso(r, s) + \WMatch(M),
\end{equation}
where $M$ is an \MWM of the complete bipartite graph on the vertex set 
$C(r) \uplus C(s)$ with weights $\wMatch(uv) = \wIso(ru,sv) + \MCSR(G^r_u, H^s_v)$ for all 
$u \in C(r)$ and $v \in C(s)$.
Hence, each edge weight corresponds to the solution of a problem of the same type 
for a pair of smaller rooted subtrees and the recursion naturally stops at the leaves.
Each subproblem, the initial one as well as those arising in recursive calls, is uniquely 
defined by a pair of rooted subtrees and essentially consists of solving a matching instance.

Figure~\ref{34_example:mcs-tree} illustrates the two rooted subtrees $G^r_u$ and $H^s_v$ and the corresponding matching problem under the weight function as given in the figure.
For rooted trees $G^r$ and $H^s$ this problem arises on the second level in the recursion of Eq. \eqref{34_eq:basic_recursion}.  We obtain $\MCSR(G^r_u, H^s_v)=\wIso(u, v)+W(M)=1+5=6$, where $M$ is an \MWM of Figure~\ref{34_example:mcs-tree}\subref{34_example:mcs-tree:mwbm}, depicted in Figure~\ref{34_example:mcs-tree}\subref{34_example:mcs-tree:mwbmsol}.

In order to compute Eq.~\eqref{34_eq:basic_recursion} the subproblems defined by the pairs of rooted subtrees 
$\mathcal{S}_\text{root}(G^r,H^s):= \{(G^r_u, H^s_v) \mid {\rm depth}(u)={\rm depth}(v) \}$ have to be solved.

\begin{proposition}\label{34_prop:rooted}
 A maximum common subtree isomorphism for two rooted trees $G^r$ and $H^s$ can be 
 computed in time $\mathcal{O}(n^3)$, where $n=|G|+|H|$.
\end{proposition}
\begin{proof}
 The bipartite graph for the subproblem $(G^r_u, H^s_v)$ contains $k_u+l_v$ vertices, where $k_u := |C(u)|$ 
 and $l_v := |C(v)|$.
 For the total running time we distinguish the cases $k_u \leq l_v$ and $k_u > l_v$. 
 For the first case we obtain a MWM in time $\mathcal{O}(k_u l_v(k_u+\log l_v))$ according to Lemma \ref{34_lem:mwmtime}. The second case is analog. Since 
 $\mathcal{S}_\text{root}(G^r,H^s) \subseteq \{(G^r_u, H^s_v) \mid u \in V_G, v \in V_H \}$
 the total running time is bounded by $\mathcal O(n^3)$ as
 \begin{equation*}
  \sum_{u \in V_G} \sum_{v \in V_H,k_u\leq l_v} k_u l_v(k_u+\log l_v) \leq
  \sum_{u \in V_G} k_u \sum_{v \in V_H} l_v(l_v+\log l_v)\leq
  n\cdot 2 n^2 \in \mathcal{O}(n^3)
 \end{equation*}
\end{proof}

\subparagraph{Unrooted trees.}
We now consider the problem for unrooted trees. An immediate solution is to solve
the rooted problem variant for all possible pairs of roots, i.e., by computing
\begin{equation}\label{34_eq:all_roots}
\textsc{MCS}(G,H) := \max \left\{\MCSR(G^r,H^s) \mid r \in V(G), s \in V(H)\right\}.
\end{equation}
Clearly, this yields the optimal solution in time $\mathcal{O}(n^5)$ with 
Proposition~\ref{34_prop:rooted}.
Note that several recursive calls involve solving the same subproblem. 
Repeated computation can easily be avoided by means of a lookup table. 
Let $\textsc{Rt}(G^r) := \{ G^r_u \mid u \in V(G) \}$ and
$\textsc{Rt}(G) := \bigcup_{r \in V(G)} \textsc{Rt}(G^r)$.
Note that we may uniquely associate the subtree $G^r_u$ with $G^p_u$, where 
$p$ is the parent of $u$ in $G^r$.
Hence, each rooted subtree $G^u_v \in \textsc{Rt}(G)$ either is the whole
tree $G$ with root $u=v$ or is the subtree rooted at $v$ of 
some edge $uv \in E(G)$, where $u$ is not contained in the subtree.
Thus, 
$\textsc{Rt}(G) = \{G^u_v \mid v \in V(G) \wedge u \in N(v) \cup \{v\} \}$ is 
the set of all rooted subtrees of $G$. In total the subproblems defined by
$\mathcal{S}(G,H):= \textsc{Rt}(G) \times \textsc{Rt}(H)$ have to be solved.

However, ensuring that each subproblem is solved only once does not 
allow to improve the bound on the running time, since $\mathcal{S}(G,H)$ still 
may contain a quadratic number of subproblems of linear size:
Let $G$ and $H$ be two star graphs on $n$ vertices, i.e., trees with all but one 
vertex of degree one. Each of the $(n-1)^2$ pairs of leaves can be selected as 
root pair and leads to a different subproblem of size $n-1$.

We show that it is sufficient to consider only a subset of the subproblems to 
guarantee that an optimal solution is found. Let 
\begin{equation}\label{34_eq:selected_roots}
\textsc{MCS}_\text{fast}(G^r,H) := \max \left\{\MCSR(G^r_u,H^s) \mid u \in V(G), s \in V(H)\right\},
\end{equation}
where $r \in V(G)$ is an arbitrary but fixed root of $G$. To compute 
Eq.~\eqref{34_eq:selected_roots}, only the subproblems 
$\mathcal{S}_\text{fast}(G^r,H) := \textsc{Rt}(G^r) \times \textsc{Rt}(H) \subseteq \mathcal{S}(G,H)$ need to be solved.

\begin{lemma}\label{34_lem:mcs_fast}
Let $\textsc{MCS}_\textnormal{fast}$ and $\textsc{MCS}$ be defined as above and  $r\in V(G)$ arbitrary but fixed, then $\textsc{MCS}_\textnormal{fast}(G^r,H) = \textsc{MCS}(G,H)$ for all trees $G$, $H$.
\end{lemma}
\begin{proof}
Let $\phi$ be an \MCSI. If $r$ is in the domain of $\phi$, then $\wIso(\phi)=\MCSR(G^r,H^{\phi(r)})=\textsc{MCS}_\textnormal{fast}(G^r,H)$.
Otherwise the domain of $\phi$ is contained in the subtree rooted at one child of $r$. Let $u$ be the unique vertex that is closest to $r$ and mapped by $\phi$. Then  $\wIso(\phi)=\MCSR(G^r_u,H^{\phi(u)})=\textsc{MCS}_\textnormal{fast}(G^r,H)$.
\end{proof}

\begin{algorithm}[t]
  \caption{Maximum Common Subtree Isomorphism}
  \label{34_alg:treemcs}
  \Input{Trees $G$ and $H$ under a weight function $\wIso$}
  \Output{Weight of an \MCSI between $G$ and $H$.}
  \Data{Table $D(u,s,v)$ storing solutions $\MCSR(G^r_u,H^s_v)$ of 
        subproblems.}
  Select an arbitrary root vertex $r\in V_G$.\;
  \ForEach(\note*[f]{All possible $G^r_u \in \textsc{Rt}(G^r)$}){$u\in V_G$ in postorder traversal on $G^r$} {\label{34_alg:treemcs:post}
    $U \gets C(u)$ in $G^r$\;
    \ForEach{$v \in V_H$} {\label{34_alg:treemcs:nodev}
      \ForEach(\note*[f]{All possible $H^s_v \in \textsc{Rt}(H)$}){$s \in N(v) \cup \{v\}$\label{34_alg:treemcs:nodes}} {
        $V \gets C(v)$ in $H^s$\; \label{34_alg:treemcs:setV}
        \uIf{$\wIso(u,v) \neq -\infty$}{
          \ForEach{pair $(u',v') \in U \times V$} {
            $\wMatch(u'v') \gets \wIso(uu',vv') + D(u',v,v')$\; \label{34_alg:treemcs:DP}
          }
          $M \gets$ \MWM of the complete graph on $U \uplus V$ with weights $\wMatch$.\; \label{34_alg:treemcs:compMWM}
          $D(u,s,v) \gets \wIso(u,v) + \WMatch(M)$\;
        } \lElse { 
          $D(u,s,v) \gets -\infty$
        }
      }
    }   
  }
  \Return the maximum entry in $D$\;
\end{algorithm}

Algorithm~\ref{34_alg:treemcs} implements this strategy, where the postorder 
traversal on $G^r$ (line~\ref{34_alg:treemcs:post}) ensures that the solutions to 
smaller subproblems are always available when required (line~\ref{34_alg:treemcs:DP}).
The lookup table contains one entry for each subproblem in $\mathcal{S}_\text{fast}(G^r,H)$ and hence requires
space $\mathcal{O}(n^2)$. %
Note that it is also possible to compute a concrete isomorphism from the \MWMs 
associated with the computed optimal solution.
The restriction of the considered subproblems allows to improve the bound on the running time.

\begin{proposition}\label{34_prop:unrooted}
 Algorithm~\ref{34_alg:treemcs} solves the maximum common subtree isomorphism problem
 for two trees $G$ and $H$ in time $\mathcal{O}(n^4)$, where $n=|G|+|H|$.
\end{proposition}
\begin{proof}
 According to Lemma~\ref{34_lem:mcs_fast} computing Eq.~\eqref{34_eq:selected_roots}
 yields the optimal solution and it suffices to solve the subproblems 
 $\mathcal{S}_\text{fast}(G^r,H)$ as realized by Algorithm~\ref{34_alg:treemcs}. 
 Let $k_u$ be the number of children of $u$ in $G^r_u$, $l^s_v$ the number of 
 children of $v$ in $H^s, s \in V(H)$, and $l_v=|N(v)|$. For all $s$ we have $l^s_v \leq l_v$.
 Similar to Proposition \ref{34_prop:rooted} the subproblems $\mathcal{S}_\text{fast}(G,H)$ can be solved in a total time of
 \begin{equation*}
   \mathcal{O}\left( \sum_{u \in V_G} \sum_{s \in V_H} \sum_{v \in V_H} (k_ul^s_v)(\min\{k_u,l^s_v\}+\log\max\{k_u,l^s_v\})\right) \subseteq
   \mathcal{O}\left(  \sum_{s \in V_H} n^3  \right) \subseteq
   \mathcal{O}\left( n^4 \right).
 \end{equation*}
\end{proof}

Further improvement of the running time is possible by no longer considering the 
MWM subroutine as a black box. We pursue this direction in the next section.
Our findings there yield the following theorem.

\begin{theorem}
\label{34_th:runtime:mcsi}
An \MCSI between two unrooted trees $G$ and $H$ can be computed in time $\mathcal{O}(|G||H| (\min \{\Delta(G),\Delta(H)\}+\log\max \{\Delta(G),\Delta(H)\}))$.
\end{theorem}
\begin{proof}
The MWM computations in Algorithm~\ref{34_alg:treemcs} are dominating, thus we obtain the above running time directly from Theorem~\ref{34_th:runtime:allmwm} of the following section.
\end{proof}

\section{Computing All Maximum Weight Matchings}
\label{34_sec:MWM}
In this section we improve the total time bound for solving all the matching
instances arising in Algorithm~\ref{34_alg:treemcs}.
First, we provide a time bound to compute an \MWM in a single bipartite graph $(V\uplus U,E)$, where possibly $|V|\neq|U|$.
In the following, we exploit the fact that during the run of our algorithm, we get sets of ``similar'' bipartite graph instances.
After computing an \MWM on one graph in one of the sets, we can derive \MWMs for all the other bipartite graphs in that set very efficiently.
Finally, we provide an upper bound to compute an \MWM in all the occurring bipartite graphs.

Computing an \MWM is closely related to finding an \MWPM and there is extensive
literature on both problems~\cite{Duan2012}.
Gabow and Tarjan~\cite{GT1989} describe a reduction to solve the \MWM problem using any algorithm for \MWPM, without altering the algorithm's asymptotic time bound, which we will make use of.
For computing an \MWPM, we use the well known Hungarian method, which has at most $n$ iterations in its outer loop and a total running time of $\mathcal O(n^3)$ or $\mathcal O(n (m + n\log n))$ using Fibonacci heaps, where $n$ and $m$ denote the number of vertices and edges of the bipartite graph.
We denote this algorithm by $\mathcal A_{\rm PM}$.

\tikzstyle{Vertex_34}=[circle, draw = black, fill=black, inner sep=0pt, minimum width=4.2pt]
\tikzstyle{VertexM_34}=[circle, draw = cyan, fill=cyan, inner sep=0pt, minimum width=4.2pt]
\tikzstyle{Edge_34} = [draw,thick,-]
\tikzstyle{EdgeM_34} = [draw=cyan,thick,-]
\tikzstyle{EdgeMM_34} = [draw=cyan,very thick,-]
\tikzstyle{EdgeA_34} = [draw=black,->]

\begin{figure}[t]
	\centering
   \begin{subfigure}[b]{0.24\textwidth}
		\begin{tikzpicture}[scale = 1.00]
			\node (a) at (0.375,1.4) [Vertex_34,label=right:\small{$v_1$}] {};
			\node (b) at (1.875,1.4) [Vertex_34,label=left:\small{$v_2$}] {};
			\node (c) at (0,0.4) [Vertex_34,label=below:\small{$u_1$}] {}
				edge [Edge_34] (a);
			\node (d) at (0.75,0.4) [Vertex_34,label=below:\small{$u_2$}] {}
				edge [Edge_34] (a);
			\node (e) at (1.5,0.4) [Vertex_34,label=below:\small{$u_3$}] {}
				edge [Edge_34] (a)
				edge [Edge_34] (b);
			\node (f) at (2.25,0.4) [Vertex_34,label=below:\small{$u_4$}] {}
				edge [Edge_34] (b);
			\node at (0.02,1) {\small{-$1$}};
			\node at (0.47,0.8) {\small{$3$}};
			\node at (1.05,1) {\small{$4$}};
			\node at (1.58,1) {\small{$2$}};
			\node at (2.15,1) {\small{$3$}};
			\node at (1.2,-1.5) {};	
		\end{tikzpicture}
		\subcaption{Input graph $B$}
		\label{34_fig:bipartiteG}
	   \end{subfigure}
	   \hfill
\begin{subfigure}[b]{0.24\textwidth}
		\begin{tikzpicture}[scale = 1.00]
			\node (a) at (0.375,1.4) [Vertex_34,label=right:\textcolor{green}{\small{$4$}}] {};
			\node (b) at (1.875,1.4) [Vertex_34,label=left:\textcolor{green}{\small{$3$}}] {};
			\node (c) at (0,0.4) [Vertex_34] {};
			\node (d) at (0.75,0.4) [Vertex_34] {}
				edge [Edge_34] (a);
			\node (e) at (1.5,0.4) [Vertex_34] {}
				edge [Edge_34] (a)
				edge [Edge_34] (b);
			\node (f) at (2.25,0.4) [Vertex_34] {}
				edge [Edge_34] (b);
			\node at (0.47,0.8) {\small{$3$}};
			\node at (1.05,1) {\small{$4$}};
			\node at (1.58,1) {\small{$2$}};
			\node at (2.17,1) {\small{$3$}};
			
			\node (ac) at (0.375,-1.4) [Vertex_34,label=right:\textcolor{green}{\small{$4$}}] {}
				edge [Edge_34, bend angle=12, bend left] (a);
			\node (bc) at (1.875,-1.4) [Vertex_34,label=left:\textcolor{green}{\small{$3$}}] {}
				edge [Edge_34, bend angle=7, bend right] (b);
			\node (cc) at (0,-0.4) [Vertex_34] {}
				edge [EdgeM_34] (c);
			\node (dc) at (0.75,-0.4) [Vertex_34] {}
				edge [Edge_34] (ac)
				edge [EdgeM_34] (d);
			\node (ec) at (1.5,-0.4) [Vertex_34] {}
				edge [Edge_34] (ac)
				edge [Edge_34] (bc)
				edge [EdgeM_34] (e);
			\node (fc) at (2.25,-0.4) [Vertex_34] {}
				edge [Edge_34] (bc)
				edge [EdgeM_34] (f);
			\node at (0.47,-0.8) {\small{$3$}};
			\node at (1.05,-1) {\small{$4$}};
			\node at (1.58,-1) {\small{$2$}};
			\node at (2.17,-1) {\small{$3$}};			
			\node at (-0.12,0) {\textcolor{cyan}{\small{$0$}}};
			\node at (0.9,0) {\textcolor{cyan}{\small{$0$}}};
			\node at (1.37,0) {\textcolor{cyan}{\small{$0$}}};
			\node at (2.38,0) {\textcolor{cyan}{\small{$0$}}};
			\node at (0.33,0) {\small{$0$}};
			\node at (1.84,0) {\small{$0$}};
		\end{tikzpicture}
		\subcaption{Reduced graph $B'$}
		\label{34_fig:initialdual}
	   \end{subfigure}
	   \hfill
		\begin{subfigure}[b]{0.24\textwidth}
		\begin{tikzpicture}[scale = 1.00]
			\node (a) at (0.375,1.4) [Vertex_34,label=right:\textcolor{green}{\small{$4$}}] {};
			\node (b) at (1.875,1.4) [Vertex_34,label=left:\textcolor{green}{\small{$3$}}] {};
			\node (c) at (0,0.4) [Vertex_34] {};
			\node (d) at (0.75,0.4) [Vertex_34] {}
				edge [Edge_34] (a);
			\node (e) at (1.5,0.4) [Vertex_34] {}
				edge [EdgeMM_34] (a)
				edge [Edge_34] (b);
			\node (f) at (2.25,0.4) [Vertex_34] {}
				edge [EdgeMM_34] (b);
			\node at (0.47,0.8) {\small{$3$}};
			\node at (1.05,1) {\textcolor{cyan}{\small{$4$}}};
			\node at (1.58,1) {\small{$2$}};
			\node at (2.17,1) {\textcolor{cyan}{\small{$3$}}};
			
			\node (ac) at (0.375,-1.4) [Vertex_34,label=right:\textcolor{green}{\small{$4$}}] {}
				edge [Edge_34, bend angle=12, bend left] (a);
			\node (bc) at (1.875,-1.4) [Vertex_34,label=left:\textcolor{green}{\small{$3$}}] {}
				edge [Edge_34, bend angle=7, bend right] (b);
			\node (cc) at (0,-0.4) [Vertex_34] {}
				edge [EdgeM_34] (c);
			\node (dc) at (0.75,-0.4) [Vertex_34] {}
				edge [Edge_34] (ac)
				edge [EdgeM_34] (d);
			\node (ec) at (1.5,-0.4) [Vertex_34] {}
				edge [EdgeM_34] (ac)
				edge [Edge_34] (bc)
				edge [Edge_34] (e);
			\node (fc) at (2.25,-0.4) [Vertex_34] {}
				edge [EdgeM_34] (bc)
				edge [Edge_34] (f);
			\node at (0.47,-0.8) {\small{$3$}};
			\node at (1.05,-1) {\textcolor{cyan}{\small{$4$}}};
			\node at (1.58,-1) {\small{$2$}};
			\node at (2.17,-1) {\textcolor{cyan}{\small{$3$}}};			
			\node at (-0.12,0) {\textcolor{cyan}{\small{$0$}}};
			\node at (0.9,0) {\textcolor{cyan}{\small{$0$}}};
			\node at (1.37,0) {\textcolor{black}{\small{$0$}}};
			\node at (2.38,0) {\textcolor{black}{\small{$0$}}};
			\node at (0.33,0) {\small{$0$}};
			\node at (1.84,0) {\small{$0$}};
		\end{tikzpicture}
		\subcaption{MWMs $M',M$}
		\label{34_fig:mwmM}
	   \end{subfigure}
	   \hfill
		\begin{subfigure}[b]{0.24\textwidth}
		\begin{tikzpicture}[scale = 1.00]
			\node (a) at (0.375,1.4) [Vertex_34,label=right:\textcolor{green}{\small{$4$}}] {};
			\node (b) at (1.875,1.4) [Vertex_34,label=left:\textcolor{green}{\small{$3$}}] {};
			\node (c) at (0,0.4) [Vertex_34] {};
			\node (d) at (0.75,0.4) [Vertex_34] {}
				edge [Edge_34] (a);
			\node (e) at (1.5,0.4) [Vertex_34] {}
				edge [EdgeM_34] (a)
				edge [Edge_34] (b);
			\node at (0.47,0.8) {\small{$3$}};
			\node at (1.05,1) {\textcolor{cyan}{\small{$4$}}};
			\node at (1.58,1) {\small{$2$}};
			
			\node (ac) at (0.375,-1.4) [Vertex_34,label=right:\textcolor{green}{\small{$4$}}] {}
				edge [Edge_34, bend angle=12, bend left] (a);
			\node (bc) at (1.875,-1.4) [Vertex_34,label=left:\textcolor{green}{\small{$3$}}] {}
				edge [Edge_34, bend angle=7, bend right] (b);
			\node (cc) at (0,-0.4) [Vertex_34] {}
				edge [EdgeM_34] (c);
			\node (dc) at (0.75,-0.4) [Vertex_34] {}
				edge [Edge_34] (ac)
				edge [EdgeM_34] (d);
			\node (ec) at (1.5,-0.4) [Vertex_34] {}
				edge [EdgeM_34] (ac)
				edge [Edge_34] (bc)
				edge [Edge_34] (e);
			\node at (0.47,-0.8) {\small{$3$}};
			\node at (1.05,-1) {\textcolor{cyan}{\small{$4$}}};
			\node at (1.58,-1) {\small{$2$}};
			\node at (-0.12,0) {\textcolor{cyan}{\small{$0$}}};
			\node at (0.9,0) {\textcolor{cyan}{\small{$0$}}};
			\node at (1.37,0) {\textcolor{black}{\small{$0$}}};
			\node at (0.33,0) {\small{$0$}};
			\node at (1.84,0) {\small{$0$}};
		\end{tikzpicture}
		\subcaption{$B'_4$ with $M_4'$}
		\label{34_fig:mwmMj}
	   \end{subfigure}
	\caption{Weighted bipartite graph $B$~(\subref{34_fig:bipartiteG});
	reduced graph $B'$ with initial duals in green (vertices without label have dual value 0) and initial matching $M''$ in blue~(\subref{34_fig:initialdual});
	\MWM $M'$ of $B'$ in blue, $M$ of $B$ in thick blue~\subref{34_fig:mwmM};
	$B'_4$ with matching $M_4'$ in blue~(\subref{34_fig:mwmMj}). - cf. proofs of Lemma \ref{34_lem:mwmtime} and Lemma \ref{34_lem:similartime}.}
	\label{34_fig:mwm}
\end{figure}

The Hungarian method is a primal-dual algorithm.
It starts with an empty matching and computes a new matching with one more edge in each iteration, maintaining a feasible dual solution of a primal linear program.
The complementary slackness theorem ensures, that the obtained perfect matching after $n$ iterations is a \MWPM. Note, by using the reduction in~\cite{GT1989}, we always have at least one perfect matching.
\begin{lemma}\label{34_lem:mwmtime}
Let $B=(V\uplus U,E)$ be a bipartite graph with edge weights $\wMatch:E\to\mathbb R$. Let $k:=|V|$, $l:=|U|,$ and $k\leq l$.
An \MWM $M$ on $B$ can be computed in time $\mathcal O(kl(k+\log l))$.
\end{lemma}
\begin{proof}
Let $\{v_1,\ldots,v_k\}=V, \{u_1, \ldots,u_l\}=U$ be the two vertex sets of $B$.
First, we remove all edges from $B$ with negative edge weight, because they never contribute to an \MWM.
Then, we add a copy $B^{\rm C}$ of $B$ to the graph.
For each vertex $v\in V\uplus U$ we denote its copy $v^{\rm C}$ and for each edge $e\in E$ we denote its copy $e^{\rm C}$.
We then copy the edge weights, i.e., $\wMatch(e^{\rm C}):=\wMatch(e)$ for each edge $e\in E$.
Next we insert a new edge of weight 0 between each vertex $v\in V\uplus U$ and its copy $v^{\rm C}$.
This graph is called \emph{reduced graph} $B'$.
Figures~\ref{34_fig:mwm}\subref{34_fig:bipartiteG} and~\subref{34_fig:initialdual} show an example of $B$ and $B'$.
An \MWPM $M'$ of $B'$ yields an \MWM $M$ of $B$: $vu\in M \Leftrightarrow vu\in M'$ and $v\in V, u\in U$. This follows from the construction of $B'$.

In the following, we prove an upper time bound to compute $M'$.
An initial feasible dual solution $d:V_{B'}\to\mathbb R$, i.e., $d(v)+d(u)\geq \wMatch(vu)$ for all edges $vu\in E_{B'}$, including $l$ matching edges $vu$ with $d(v)+d(u)= \wMatch(vu)$, is computed as follows:
We set $d(u)=0$ for all $u\in U$ and $d(v):=\max\{\wMatch(vu) \mid u\in U\}$ for all $v\in V$.
Next, for each $v\in V\uplus U$ the vertex $v^{\rm C}$ obtains the dual value $d(v^{\rm C}):=d(v)$.
We define an initial matching $M'':=\{uu^{\rm C}\mid u\in U\}$. Note, $d(u)+d(u^{\rm C})=0=\wMatch(uu^{\rm C})$. 

The dual solution $d$ is feasible and can be computed in time $\mathcal O(kl)$.
Let $n:=|V_{B'}|=2(k+l)\in \Theta(l)$ and $m:=|E_{B'}|\leq 2kl+k+l\in \mathcal O(kl).$
Increasing the number of matching edges by one using a single iteration of $\mathcal A_{\rm PM}$ is possible in time $\mathcal O(m+n\log n)=\mathcal O(l(k+\log l))$.
To obtain an \MWPM $M'$ form $M''$ in $B'$ we need to increase the number of matching edges by $k$, therefore the time to compute $M'$ and thus $M$ is $\mathcal O(kl(k+\log l))$.
\end{proof}

\begin{algorithm}[t]
  \caption{Computing MWMs on $B$ and $B_j$, cf. Lemmas \ref{34_lem:mwmtime}, \ref{34_lem:similartime} }
  \label{34_alg:lemma5_6}
  \Input{Bipartite graph $B=(V\uplus U,E), |U| \geq 2, U=\{u_1,u_2,\ldots\}$ with edge weights $\wMatch:E\to\mathbb R$}
  \Output{MWMs $M,M_j$ on $B, B_j:=G[V\uplus U\setminus \{u_j\}]$ for each $j\in [1..|U|]$.}
  \uIf(\note*[f]{Compute MWM $M$ of $B$}){$|V| \leq |U|$} {
	Let $B':=(V',E')$, where $V':=V\cup U \cup \{v^{\rm C}\mid v\in V\cup U\}$ and $E':=E\cup \{e^{\rm C}\mid e\in E\}\cup \{vv^{\rm C}\mid v\in V\cup U\}$.\;\label{34_alg:lemma5_6:startif}
	$\wMatch(e^{\rm C})\gets\wMatch(e)$ for all $e\in E$\note*[f]{Weights of additional edges}\\
	$\wMatch(vv^{\rm C})\gets 0$ for all $v\in V\cup U$\;
	$d(u^{\rm C})\gets d(u)\gets 0$ for all $u \in U$ \note*[f]{Dual values}\\
	$d(v^{\rm C})\gets d(v)\gets\max\{\wMatch(vu) \mid u\in U\}$ for all $v\in V$\;
	$M''\gets\{uu^{\rm C}\mid u\in U\}$ \note*[f]{Initial matching edges}\\
	Starting with $M''$ and $d$, compute an MWPM $M'$ on $B'$ using $|V|$ iterations of $\mathcal A_{\rm PM}$\;
	$M\gets \{vu\mid vu\in M', v\in V, u\in U\}$ \; \label{34_alg:lemma5_6:endif}
  } \Else { 
    Exchange the vertices of $V$ and $U$.\;\label{34_alg:lemma5_6:exchange1}
    	Compute $M$ as in lines \ref{34_alg:lemma5_6:startif} to \ref{34_alg:lemma5_6:endif} and exchange $V$ and $U$ back.\label{34_alg:lemma5_6:exchange2}
  }	
  $d\gets$ The dual values obtained while computing $M'$.\\
 \ForEach(\note*[f]{MWMs $M_j$ on $B_j$}){$j \in [1..|U|]$} {
   \uIf{$u_j$ is not matched by $M$}{
     $M_j\gets M$
   }\Else { 
     $B'_j\gets B'\setminus \{u_j,u_j^{\rm C}\}$\\
     $M_j'\gets M'$ without the matching edges incident to $u_j,u_j^{\rm C}$\note*[f]{Initial matching}\\
     Compute an MWPM $M_j'$ on $B'_j$ using $d$ and a single iteration of $\mathcal A_{\rm PM}$.\\
     $M_j\gets \{vu\mid vu\in M_j', v\in V, u\in U\}$ \\
   }
 }
\end{algorithm}

\begin{lemma}\label{34_lem:similartime}
Let $B=(V\uplus U,E)$ be a weighted bipartite graph with $k:=|V|, U=\{u_1,\ldots,u_l\}, l \geq 2$. Let $B_j:=G[V\uplus U\setminus \{u_j\}]$ for each $j\in [1..l]$.
Computing \MWMs for all graphs $B,B_1,\ldots,B_l$ is possible in total time 
$\mathcal O(kl(\min\{k,l\}+\log \max\{k,l\}))$.
\end{lemma}

\begin{proof}
According to Lemma \ref{34_lem:mwmtime} we obtain an \MWM $M$ of $B$ in time $\mathcal O(kl(\min\{k,l\}+\log \max\{k,l\}))$.
We compute an \MWM on each $B_j$ as follows:
Let $d$ be an optimal dual solution obtained while computing $M'$ (on $B'$, see proof of Lemma  \ref{34_lem:mwmtime}).
If $u_j$ is not matched by $M$, i.e., $u_j\notin e$ for all $e\in M$, then $M_j:=M$ is an \MWM of $B_j$.
Otherwise let $B'_j$ be the reduced graph as explained in the proof of Lemma \ref{34_lem:mwmtime}.
We obtain a feasible dual solution $d_j$ on the bipartite graph $B'_j$ by taking the dual values from $d$, i.e., $d_j(v):=d(v)$ for all $v\in V(B'_j)$.
Note, we have $2(k+l)$ vertices in $B'$, and exactly two less in $B'_j$, i.e., a perfect matching in $B'_j$ consists of $k+l-1$ matching edges.

We can derive an initial matching $M_j'$ on $B'_j$ with $k+l-2$ edges from $M'$;
$M'_j$ contains the matching edges that are not incident to the two removed vertices from $B'$ to $B_j'$.
Therefore only one more iteration of $\mathcal A_{\rm PM}$ is needed, which is possible in time $\mathcal O(\max\{k,l\}(\min\{k,l\}+\log \max\{k,l\}))$, cf. proof of Lemma \ref{34_lem:mwmtime}.
We then obtain $M_j$ from $M_j'$ as previously described.
The complementary slackness conditions ensure $M_j'$ and therefore $M_j$ is of maximum weight.
An example of $M'$ and $B'_j\ (j=4)$ is shown in Figures~\ref{34_fig:mwm}\subref{34_fig:mwmM} and~\subref{34_fig:mwmMj}.

We need to compute an \MWM different from $M$ for at most $\min\{k,l\}$ of the $l$ graphs $B_1,\ldots,B_l$, because at most $k$ vertices of $U$ are matched by $M$, cf. Figure~\ref{34_fig:mwm}\subref{34_fig:mwmM}: only $u_3$ and $u_4$ of $U$ are matched by $M$.
Therefore the time bound to compute \MWMs for all the graphs $B_1,\ldots,B_l$ is $\mathcal O(\min\{k,l\} \max\{k,l\}(\min\{k,l\}+\log \max\{k,l\}))=\mathcal O(kl(\min\{k,l\}+\log \max\{k,l\}))$.
\end{proof}

We call $B$ and the graphs $B_j,j\in [1..l]$,  a set of ``similar'' bipartite graph instances.
Algorithm~\ref{34_alg:lemma5_6} shows how we compute an MWM for each graph in this set.
Next, we apply Lemma~\ref{34_lem:similartime} to Algorithm~\ref{34_alg:treemcs}. For each pair $u\in V_G,v\in V_H$ of vertices, selected in line \ref{34_alg:treemcs:post} and \ref{34_alg:treemcs:nodev}, respectively, the algorithm computes up to $|N(v)|+1$ \MWMs, cf. lines \ref{34_alg:treemcs:nodes}, \ref{34_alg:treemcs:compMWM}. A close look at Algorithm~\ref{34_alg:treemcs} reveals this as a set of ``similar'' bipartite graph instances. The first graph $B$ is obtained by selecting $s=v$ in line \ref{34_alg:treemcs:nodes}. The other graphs $B_j$ are obtained by selecting all the vertices $s\in N(v)$. 
This observation allows to prove the following theorem.

\begin{theorem}
\label{34_th:runtime:allmwm}
All the \MWMs in Algorithm~\ref{34_alg:treemcs} can be computed in total time 
\\$\mathcal{O}(kl (\min \{\Delta(G),\Delta(H)\}+\log\max \{\Delta(G),\Delta(H)\}))$, where $k=|G|$ and $l=|H|$.
\end{theorem}

\begin{proof}
For each pair $(v,u)\in V_{G}\times V_H$ we compute an \MWM on each of the ``similar'' graphs, where $B=(C(v)\uplus N(u),E)$ and edge weights as determined by Eq.~\eqref{34_eq:basic_recursion}.
Let $d_{\min}:=\min \{\Delta(G),\Delta(H)\}$ and $d_{\max}:=\max \{\Delta(G),\Delta(H)\}$.
For all the pairs $(v,u)$ we obtain a time complexity of
\begin{align*}
&\mathcal O\left(\sum_v\sum_u \delta(v) \delta(u)(\min \{\delta(v),\delta(u)\}+\log\max \{\delta(v),\delta(u)\})\right)\\
\subseteq\ &\mathcal O\left(\sum_v \delta(v) \sum_u \delta(u)(d_{\min}+\log d_{\max})\right)\\
=\ &\mathcal O\left((d_{\min}+\log d_{\max}) \sum_v \delta(v) l\right)\\ 
=\ &\mathcal O((d_{\min}+\log d_{\max}) kl).
\end{align*} 
\end{proof}

\section{Lower Bounds on the Time Complexity and Optimality}
\label{34_sec:lowerbound}
Providing a tight lower bound on the time complexity of a problem is generally a 
non-trivial task.
We obtain this for trees of bounded degree and reason why the existence of an
algorithm with subcubic running time for unrestricted trees is unlikely.
In order to solve MCST with an arbitrary weight function $\wIso$ for two trees 
$G$ and $H$, all values $\wIso(u,v)$ for $u \in V(G)$ and $v \in V(H)$ must be
considered. This directly leads the lower bound of $\Omega(|G||H|)$ for the time
complexity of MCST.
For trees of bounded degree our approach achieves running time 
$\mathcal{O}(|G||H|)$ according to Theorem~\ref{34_th:runtime:mcsi} and, thus, 
has an optimal worst-case running time in the considered setting.

For unrestricted trees of order $n$ our approach has a running time of 
$\mathcal O(n^3)$ according to Theorem~\ref{34_th:runtime:mcsi}.
In the next paragraph we present a linear time reduction from the assignment problem to MCST, which preserves the time complexity. Therefore solving MCST in time $o(n^3)$ yields an algorithm to solve the assignment problem in time $o(n^3)$.
The Hungarian method solves the assignment problem in $\mathcal O(n^3)$, which is the best known time bound for bipartite graphs with $\Theta(n^2)$ edges of unrestricted weight for more than 30 years. 

Let $B=(U \uplus V, E, \wMatch)$ be a weighted bipartite graph on which we want to solve the assignment problem, i.e., to find an \MWPM. 
We assume weights to be non-negative, which can be achieved by adding a sufficiently large constant to every edge weight to obtain an assignment problem
that is equivalent w.r.t.\ the \MWPMs.
We construct a star graph $G$ with center $c$ and leaves $U$ and another star
graph $H$ with center $c'$ and leaves $V$.
Let $n=|U|=|V|$ and $N = \max_{e \in E} \wMatch(e)$.
We define $\wIso$ such that $\wIso(u,v)=\wMatch(uv)+nN$ for all $uv\in E$, $\wIso(c,c')=nN$ and $\wIso(u,v)=-\infty$ for all other pairs of vertices.
For all pairs of edges we define $\wIso(e,e')=0$.
Let $\phi$ be an \MCSI between $G$ and $H$ w.r.t.\ $\wMatch$ and $p:=|\dom(\phi)|$.
It directly follows from the construction that 
$M := \{uv \in E \mid \phi(u) = v \}$ is an \MWM in $B$ with 
$\WMatch(M) = \WIso(\phi) - pnN$.
Furthermore, the incremented weights ensure that $M$ is perfect, i.e., $p = n+1$, whenever $B$ admits a perfect matching.
Therefore we obtain:

\begin{proposition}
\label{34_th:lowerbound:unrestricted}
Only if we can solve the assignment problem on a graph with $n$ vertices and $\Theta(n^2)$ edges of unrestricted weight in time $o(n^3)$, we can solve MCST
on two unrooted trees of order $\Theta(n)$ in time $o(n^3)$.
\end{proposition}

\section{Output-Sensitive Algorithms for Listing All Solutions}
Algorithm~\ref{34_alg:treemcs} can easily be modified to not only output the weight
of an \MCSI, but also an associated isomorphism. Let $D(u,s,v)$ be a maximum entry in $D$. Then $\phi(u)=v$. Further mappings are defined by the matching edges occurring in Eq.~\eqref{34_eq:basic_recursion}. In the example of Figure~\ref{34_example:mcs-tree}\subref{34_example:mcs-tree:mwbmsol} we obtain $\phi(c_1)=d_1$ and $\phi(c_2)=d_3$.
Since in general there is no single unique \MCSI, it is of interest to find and 
list all of them.
In this section we show how our techniques can be combined with the enumeration
algorithm from~\cite{DHKM14}, which lists all the different \MCSIs of two trees exactly once. We obtain the best known time bound for listing all
solutions by an improved analysis.

Since the number of MCSIs is not polynomially bounded in the size of the input 
trees, we cannot expect polynomial running time. An algorithm is said to be
\emph{output-sensitive} if its running time depends on the size of the output in 
addition to the size of the input. 

The basic idea to enumerate all \MCSIs is to first compute the weight of an \MCSI.
Then for each maximum table entry $D(u,v,v), u\in V_G, v\in V_H$, all the different rooted \MCSIs on the rooted subtrees $G^r_u,H^v_v$ are listed.
Note, we omit maximum table entries $D(u,s,v)$, where $s\neq v$. We do this, because every \MCSI of $G^r_u,H^s_v$ is also an \MCSI of $G^r_u,H^v_v$.
As an example let $u$ be the root of $G$ in Figure \ref{34_example:mcs-tree}.
Then $D(u,v,v)=D(u,s,v)=7$. For both table entries we obtain the same \MCSI $\phi$ with $\phi(u)=v, \phi(r)=d_1, \phi(c_1)=d_2, \phi(c_2)=d_3,\ldots$.

We enumerate the \MCSIs on a pair of rooted subtrees by enumerating all \MWMs of the associated bipartite graphs of Eq.~\eqref{34_eq:basic_recursion}
 and then expanding $\phi$ recursively along all the different \MWMs of the mapped children.
For the problem depicted in Figure~\ref{34_example:mcs-tree}\subref{34_example:mcs-tree:mwbm} there are two different \MWMs:
 $M_1=\{c_1d_1,c_2d_3\}$ and $M_2=\{c_1d_2,c_2d_3\}$. 
Therefore we first expand along $M_1$ as explained in the first paragraph of this section and then along $M_2$. 
We do this recursively for each occurring matching instance. 
The enumerated isomorphisms of each maximum entry are pairwise different, based on the different \MWMs.
They are also pairwise different between two different maximum entries. The proof of the latter claim is similar to the proof of Lemma~\ref{34_lem:mcs_fast}.
Thus we do not enumerate an \MCSI twice.
Further we do not omit an \MCSI, because we consider all necessary maximum table entries and their rooted subtrees,
 as well as all possible expansions along the \MWMs.

Note, the enumeration algorithm of~\cite{DHKM14} uses a somewhat different table to store maximum solutions. The basic idea to list all solutions is the same.
For trees of sizes $k:=|G|$ and $l:=|H|$, $k\leq l$, their enumeration algorithm 
requires total time $\mathcal O(k^2l^4+Tl^2)$, where $T$ is the number of different \MCSIs.
The $\mathcal O(k^2l^4)$ term of the running time is caused by computing the weight of an \MCSI in time $\mathcal O(kl^4)$ and repeated deletions of single edges in one tree and recalculations of the weight of an \MCSI to avoid outputting an \MCSI twice.
We have improved the time bound to compute the weight of an \MCSI, cf.~Theorem~\ref{34_th:runtime:mcsi}.
Therefore we can improve the $\mathcal O(k^2l^4)$ term to $\mathcal O(kl(\min \{\Delta(G),\Delta(H)\}+\log\max \{\Delta(G),\Delta(H)\}))$.

The $\mathcal O(Tl^2)$ term in the original running time is caused by the enumeration of \MWMs.
For each \MCSI $\phi$ several \MWMs have to be enumerated, let this number be $m_{\phi}$.
The time to do this can be bounded by $\mathcal O(l^2)$, when using a variant of the enumeration algorithm for perfect matchings presented in~\cite{Uno97}.
The running time follows from the fact, that for each \MWM two depth first searches (DFS) in a directed subgraph of $B'$,  cf.~Figure~\ref{34_fig:mwm}, are computed.
The running time of DFS is linear in the number of edges and vertices. 
Let $k_i,l_i$ be the sizes of the disjoint vertex sets of the $i$-th bipartite graph, on which we enumerate the \MWMs, $i\in [1..m_{\phi}]$.
Then $\sum_i k_i\leq k$ and $\sum_i l_i\leq l$, because all the vertices in all the $m_{\phi}$ bipartite graphs are pairwise disjoint.
The running time of DFS in the directed subgraphs of the $i$-th bipartite graph is $\mathcal{O}(k_il_i)$, cf. Figure~\ref{34_fig:mwm}\subref{34_fig:initialdual} or~\subref{34_fig:mwmMj}.
For all $m_{\phi}$ DFS runs we have $\sum_i k_il_i\leq \sum_i k_i\Delta(H)\leq k\Delta(H)$ as well as $\sum_i k_il_i\leq \sum_i \Delta(G)l_i\leq \Delta(G) l$.
Hence, the time to enumerate $\phi$ is bounded by $\mathcal O(\min\{k\Delta(H),\Delta(G)l\})$.

Both improvements combined together, the initial computation of the weight of an \MCSI and   the \MWM enumeration, improve the enumeration time from $\mathcal{O}(n^6+Tn^2)$ to $\mathcal{O}(n^3+Tn\min\{\Delta(G),\Delta(H)\})$. More precisely we obtain the following theorem.

\begin{theorem}\label{34_th:runtime:allmcsi}
Enumerating all \MCSIs of two unrooted trees $G$ and $H$ is possible in time\\ 
$\mathcal{O}(\,|G|\,|H|\,(\,\min \{\Delta(G),\Delta(H)\}\,+\,\log\max \{\Delta(G),\Delta(H)\}\,)+T(\,\min\{|G|\Delta(H),\Delta(G)|H|\}\,)\,)$,\\
where $T$ is the number of different \MCSIs.
\end{theorem}

\section{Experimental Comparison}
In this section we experimentally evaluate the running time of our approach (DKM)
on synthetic and real-world instances. 
We compare our algorithm to the approach of~\cite{Schietgat2013} which also 
solves MCST when the input graphs are trees. The corrected analysis of the approach 
yields a running time of  $\mathcal{O}(n^{4})$, see Appendix,
which aligns better with our experimental findings of $\Omega(n^5)$.
The implementation was provided by the authors as part of the FOG package.\footnote{https://dtai.cs.kuleuven.be/software/PMCSFG}
Both algorithms were implemented in C\texttt{++} and compiled with GCC v.4.8.4.
Running times were measured on an Intel Core i7-3770 CPU with 16 GB of RAM
using a single core only.
We generated random trees by iteratively adding edges to a randomly chosen vertex
and averaged over 40 to 100 pairs of instances depending on their size.
The weight function $\wIso$ was set to 1 for each pair of vertices and edges, i.e., we compute isomorphisms of maximum size. This matches the setting in FOG.

Table~\ref{34_tab:results} summarizes our results and we observe that the running
time of our approach aligns with our theoretical analysis.
In comparison, FOG's running time is much higher and increases to a larger extent with the input size.
The running times of both algorithms show a low standard deviation for random trees,
cf. Tables~\ref{34_tab:results}\subref{34_tab:results:same},~\subref{34_tab:results:fixed}.
Table~\ref{34_tab:results}\subref{34_tab:results:star} shows the running time in star graphs, which are
worst-case examples for some approaches, see Sec.~\ref{34_sec:FDPF}.
Our theoretical proven cubic running  time matches the experimental results, while FOG's running time increases drastically.
Table~\ref{34_tab:results}\subref{34_tab:results:labeled} summarizes the computation time under different weight functions.   We defined $\wIso$ such that different labels are simulated, i.e., vertices and edges with different labels have weight $-\infty$, which again matches FOG's setting. Both algorithms clearly benefit from the fact that less \MWMs have to be computed. The results on random trees are also shown in Figure~\ref{34_fig:exp1}.

From a chemical database of thousands of molecules\footnote{%
NCI Open Database, GI50, \url{http://cactus.nci.nih.gov}}
we extracted 100 pairs of graphs with block-cut trees (BC-trees) consisting of more than 40 vertices. BC-trees are a representation of graphs, where each maximal biconnected component is represented by a $B$-vertex. If two such components share a vertex, the corresponding $B$-vertices are connected through a $C$-vertex representing this shared vertex.
The running time for MCST on BC-trees is an important factor for the total running time of MCS algorithms for outerplanar and series-parallel graphs like~\cite{Akutsu2013,Kriege2014a,Schietgat2013}.
The average running time of our algorithm was 11.2\,ms, compared to FOG's 481.3\,ms. The speedup factor ranges from 24 to 59, with an average of 43. This indicates that the above mentioned approaches could greatly benefit from the techniques presented in this paper.

\begin{table}[t]
\centering
\begin{subtable}[b]{0.495\textwidth}
\begin{tabular}{|c||r|r|r|}
\hline
Order&DKM&FOG&$q$\\
\hline\hline
20&$0.9\pm8\%$&$40\pm7\%$&44.1\\
\hline
40&$3.5\pm6\%$&$221\pm5\%$&62.7\\
\hline
80&$15.2\pm4\%$&$1\,286\pm5\%$&84.8\\
\hline
160&$58.9\pm3\%$&$8\,342\pm5\%$&141.7\\
\hline
320&$237.4\pm2\%$&$63\,327\pm8\%$&266.9\\
\hline
\end{tabular}
\subcaption{Random trees of the same order}
\label{34_tab:results:same}
\end{subtable}
\hfill
\begin{subtable}[b]{0.495\textwidth}
\begin{tabular}{|c||r|r|r|}
\hline
\phantom{Q}$|H|$\phantom{Q}&DKM&FOG&$q$\\
\hline\hline
20&$3.6\pm8\%$&$192\pm4\%$&53.6\\
\hline
40&$7.3\pm7\%$&$504\pm4\%$&68.7\\
\hline
80&$15.2\pm4\%$&$1\,286\pm5\%$&84.8\\
\hline
160&$30\pm9\%$&$3080\pm4\%$&103.3\\
\hline
320&$59.5\pm3\%$&$6842\pm4\%$&114.9\\
\hline
\end{tabular}
\subcaption{Random trees with $|G|=80$ fixed}
\label{34_tab:results:fixed}
\end{subtable}
\hfill
\begin{subtable}[b]{0.495\textwidth}
\begin{tabular}{|c||r|r|r|}
\hline
Order&DKM&\hspace{2.5em}FOG&\hspace{1.7em}$q$\\
\hline\hline
10&$0.1$&$18$&117.6\\
\hline
20&$1$&$489$&458.5\\
\hline
40&$8.9$&$18\,722$&2109.9\\
\hline
80&$77.5$&$929\,784$&11\,992.1\\
\hline
\end{tabular}
\subcaption{Star graphs}
\label{34_tab:results:star}
\end{subtable}
\hfill
\begin{subtable}[b]{0.495\textwidth}
\begin{tabular}{|c||r|r|r|}
\hline
\#labels&DKM&\hspace{3.1em}FOG&\hspace{.8em}$q$\\
\hline\hline
1&$15.2\pm4\%$&$1\,286\pm5\%$&84.8\\
\hline
2&$5.4\pm8\%$&$217\pm8\%$&40\\
\hline
3&$3.3\pm7\%$&$118\pm12\%$&36.1\\
\hline
4&$2.6\pm8\%$&$83\pm9\%$&31.9\\
\hline
\end{tabular}
\subcaption{Different $\wIso$-functions, order $80$}
\label{34_tab:results:labeled}
\end{subtable}
\caption{Average running time in ms $\pm$ RSD in \% and speedup factor $q:=$ FOG/DKM.}
\label{34_tab:results}
\end{table}

\begin{figure}[t]
\def\scax_34{*0.05}
\def\scay_34{*0.006}
		\begin{tikzpicture}[scale = 1.00]
			\node (xr) at (110\scax_34,0) [label=right:\hspace*{-2ex}\small{Order}] {};
			\node (xl) at (-5\scax_34,0) {}
				edge[EdgeA_34] (xr);
			\node (ytt) at (5\scax_34,450\scay_34+0.1) [] {\small{Time in ms}};
			\node (yt) at (5\scax_34,450\scay_34) [] {};
			\node (yb) at (5\scax_34,-30\scay_34) {}
				edge[EdgeA_34] (yt);
			\node () at (10\scax_34,0) [label=below:\small{10}] {};
			\node () at (20\scax_34,0) [label=below:\small{20}] {};
			\node () at (30\scax_34,0) [label=below:\small{30}] {};
			\node () at (40\scax_34,0) [label=below:\small{40}] {};
			\node () at (50\scax_34,0) [label=below:\small{50}] {};
			\node () at (60\scax_34,0) [label=below:\small{60}] {};
			\node () at (70\scax_34,0) [label=below:\small{70}] {};
			\node () at (80\scax_34,0) [label=below:\small{80}] {};
			\node () at (90\scax_34,0) [label=below:\small{90}] {};
			\node () at (100\scax_34,0) [label=below:\small{100}] {};
			\node (a) at (1\scax_34,100\scay_34) [label=left:\small{100}\hspace*{-2ex}] {};
			\node () at (9\scax_34,100\scay_34) {}
			   edge [Edge_34] (a);
			\node (a) at (1\scax_34,200\scay_34) [label=left:\small{200}\hspace*{-2ex}] {};
			\node () at (9\scax_34,200\scay_34) {}
			   edge [Edge_34] (a);
			\node (a) at (1\scax_34,300\scay_34) [label=left:\small{300}\hspace*{-2ex}] {};
			\node () at (9\scax_34,300\scay_34) {}
			   edge [Edge_34] (a);
			\node (a) at (1\scax_34,400\scay_34) [label=left:\small{400}\hspace*{-2ex}] {};
			\node () at (9\scax_34,400\scay_34) {}
			   edge [Edge_34] (a);

			\node (z) at (10\scax_34,6\scay_34)[VertexM_34] {};
			\node (a) at (20\scax_34,40\scay_34)[VertexM_34] {}
				edge [EdgeM_34] (z);
			\node (b) at (30\scax_34,109\scay_34)[VertexM_34] {}
				edge [EdgeM_34] (a);
			\node (c) at (40\scax_34,219\scay_34)[VertexM_34] {}
				edge [EdgeM_34] (b);
			\node (d) at (50\scax_34,387\scay_34)[VertexM_34] {}
				edge [EdgeM_34] (c);
				
			\node (z) at (10\scax_34,0.2\scay_34)[Vertex_34] {};
			\node (a) at (20\scax_34,0.9\scay_34)[Vertex_34] {}
				edge [Edge_34] (z);
			\node (b) at (30\scax_34,2.1\scay_34)[Vertex_34] {}
				edge [Edge_34] (a);
			\node (c) at (40\scax_34,3.5\scay_34)[Vertex_34] {}
				edge [Edge_34] (b);
			\node (d) at (50\scax_34,6.0\scay_34)[Vertex_34] {}
				edge [Edge_34] (c);
			\node (e) at (60\scax_34,8.7\scay_34)[Vertex_34] {}
				edge [Edge_34] (d);
			\node (y) at (70\scax_34,11.8\scay_34)[Vertex_34] {}
				edge [Edge_34] (e);
			\node (f) at (80\scax_34,15.2\scay_34)[Vertex_34] {}
				edge [Edge_34] (y);
			\node (x) at (90\scax_34,19.3\scay_34)[Vertex_34] {}
				edge [Edge_34] (f);
			\node (g) at (100\scax_34,23.8\scay_34)[Vertex_34] {}
				edge [Edge_34] (x);
		\end{tikzpicture}
		\caption{Average running time in ms (y-axis) for \MCSI computation on random trees of order $n$ (x-axis). Black = Our implementation (DKM). Blue = FOG implementation.}
		\label{34_fig:exp1}
	   \end{figure}
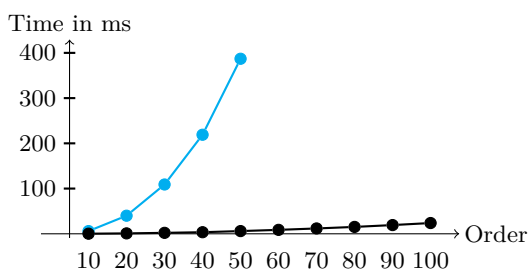

\section{Conclusions}
We have presented a novel algorithm for MCST which 
(i) considers only the subproblems required to guarantee that an optimal 
solution is found and 
(ii) solves groups of related matching instances efficiently in one pass. 
Rigorous analysis shows that the approach achieves cubic time in general trees and 
quadratic time in trees of bounded degree. Our analysis of the problem complexity
reveals that there is only little room for possible further improvements. The 
practical efficiency is documented by an experimental comparison. 

If the weight function is restricted to integers of a bounded value, scaling
approaches~\cite{Duan2012} to the corresponding matching problems become applicable.
It remains future work to improve the running time for this case.

\bibliography{lit}
\appendix
\setboolean{@twoside}{false}


\section*{Supplementary material}

# Erratum to: A polynomial-time maximum common subgraph algorithm for outerplanar graphs and its application to chemoinformatics

Andre Droschinsky    Nils M. Kriege    Petra Mutzel

February 10, 2016

## Abstract

In this erratum it is shown that the article by Schietgat, Ramon, and Bruynooghe [2013] contains the following errors:

(i) The running time of the presented algorithm for computing a maximum common subgraph under block-and-bridge-preserving subgraph isomorphism (BBP-MCS) for two graphs of order $n$ is $\Omega(n^4)$ instead of $\mathcal{O}(n^{2.5})$ as claimed in the article.

(ii) The dissimilarity measure derived from BBP-MCS does not always fulfill the triangle inequality and, hence, is not a metric.

As a consequence, still no algorithm with running time $o(n^4)$ is known, neither for BBP-MCS nor for the special case, where both input graphs are trees. The dissimilarity measure should not be used in combination with techniques that rely on or exploit the triangle inequality in any way.

## 1  Introduction

The polynomial-time algorithm presented by Schietgat et al. [2013] solves a variation of the maximum common subgraph problem in outerplanar graphs and applies the algorithm to molecular graphs. We assume the terminology and notation introduced in the original article and directly refer to the numbered definitions, theorems and algorithms therein.

## 2  Complexity Analysis

Theorem 2 states that for two outerplanar graphs $G$ and $H$ the proposed BBP-MCS algorithm runs in time

$$\mathcal{O}\left(|V(G)| \cdot |V(H)| \cdot (|V(G)| + |V(H)|)^{\frac{1}{2}}\right),$$

which is $\mathcal{O}(n^{2.5})$ for $|V(G)| = |V(H)| = n$. We show that this bound cannot be obtained by the presented techniques.

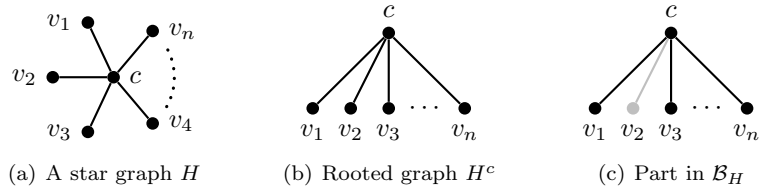

| (a) A star graph $H$ | (b) Rooted graph $H^c$ | (c) Part in $\mathcal{B}_H$ |

Figure 1: A star graph of order $n$ (a), the star graph rooted at the center vertex (b) and an elementary part $H^c \setminus \{v_2\}$ (c) obtained from $H^{v_2}$, where the gray vertex with its incident edge is deleted.

## 2.1 Solving Weighted Maximal Matching Problems

The algorithm makes use of a subroutine for solving the weighted maximal matching problem in bipartite graphs, where weights are real values, cf. Definition 16. The matching instances arising in the course of the algorithm are considered to be complete bipartite graphs, see proof of Theorem 2. Hence, the running times given in the following refer to complete bipartite graphs with $n$ vertices in order to improve readability. The authors propose to use the algorithm by Hopcroft and Karp [1973] to solve an instance of the problem in time $\mathcal{O}(n^{2.5})$. Since this algorithm computes a matching of maximal cardinality, but is not designed to take weights into account, it cannot be applied to the instances that occur. The best known algorithm for the problem is the well-known Hungarian method, which requires time $\mathcal{O}(n^3)$ [Papadimitriou and Steiglitz, 1982]. Note that by assuming that weights are integers within the range of $[0..N]$ scaling algorithm would become applicable like [Duan and Su, 2012], which solve the problem in $\mathcal{O}(n^{2.5} \log N)$. However, this would require to alter the basic definition of the problem, which states that the size of a common subgraph is a real number determined by the labels of the vertices and edges it contains, cf. Definition 2.

In summary, no better time bound than $\mathcal{O}(n^3)$ on the subproblem of solving weighted maximal matchings with $n$ vertices can be assumed.

## 2.2 The Number of Matching Instances

We consider a particularly simple counterexample to illustrate the flaw regarding the analysis of the running time in the proof of Theorem 2. We show that for two graphs $G$ and $H$ of order $\Theta(n)$ the presented algorithm performs $\Theta(n)$ calls to the subroutine for weighted maximal matching (MAXMATCH, Algorithm 2) with instances of size $\Theta(n)$. Assuming cubic running time for each matching problem, cf. Section 2.1, this yields a lower bound on the total running of the presented algorithm of $\Omega(n^4)$, contradicting Theorem 2.

Let the two graphs $G$ and $H$ both be star graphs of order $n+1$, i.e., trees with all but one vertex of degree one as depicted in Figure 1(a). Since trees in particular are outerplanar, $G$ and $H$ are valid input graphs for BBP-MCS. The presented algorithm greatly simplifies: Algorithm 4 and Algorithm 3, lines 11-18, will not be required to solve the problem on trees. The algorithm relies on a decomposition of the two input graphs into their parts. Let $L(G)$ denote the leafs of a graph $G$. Following Definitions 20, 23, 26 we select the center

vertex of $G$ as distinguished root $r$ and obtain the parts

$$Parts(G^r) = \underbrace{\{G^r\}}_{\mathcal{F}_G} \cup \underbrace{\{(\{v\}, \emptyset)^v \mid v \in L(G)\}}_{\mathcal{L}_G}, \quad (1)$$

where $\mathcal{L}_G$ are rooted subgraphs containing a single leaf of $G$, which are the elementary parts of the compound-root BPS $G^r$. With Definition 27 we decompose $H$ into the parts

$$Parts^*(H) = \underbrace{\{H^r \mid r \in V(H)\}}_{\mathcal{F}_H} \cup \underbrace{\{(\{v\}, \emptyset)^v \mid v \in L(H)\}}_{\mathcal{L}_H} \cup \underbrace{\{H^c \setminus \{v\} \mid v \in L(H)\}}_{\mathcal{B}_H}, \quad (2)$$

where $c$ is the unique center vertex of $H$ and $\mathcal{B}_H$ the subgraphs rooted at $c$ obtained by deleting a single leaf with its incident edge, cf. Figure 1(c). Note that the elements in $\mathcal{F}_G$ and $\mathcal{B}_H$ are compound-root graphs. Consequently $\mathcal{F}_G \times \mathcal{B}_H \subseteq Pairs(G, H)$ according to Definition 28 and Algorithm 1 calls the subprocedures RMCS and RMCSCOMPOUND for each pair $(P, Q) \in \mathcal{F}_G \times \mathcal{B}_H$. The procedure RMCSCOMPOUND, see Algorithm 2, constructs and solves a weighted maximal matching instance of size $a + b$, where $a$ and $b$ are the number of elementary parts of $P$ and $Q$, respectively. The number of elementary parts of $P$ in $\mathcal{F}_G$ is $|\mathcal{L}_G| = n$, the number of elementary parts of each $Q$ in $\mathcal{B}_H$ is $n - 1$. Hence, each of these matching instances has size $2n - 1$ and requires cubic time. The number of such pairs is $|\mathcal{F}_G \times \mathcal{B}_H| = n$. Therefore, the total running time of the algorithm is $\Omega(n^4)$ in contradiction to Theorem 2.

Consequently, there must be an error in the proof of Theorem 2: In the proof it is claimed that every vertex $g \in V(G)$ and every vertex $v \in V(H)$ has at most $\deg(v)$ (resp. $\deg(h)$) elementary parts involved in a maximal matching [Schietgat et al., 2013]. While this statement is correct the subsequent analysis does not take into account that there may be up to $\deg(v)$ matching instances of that size for a vertex $v \in V(H)$ as seen from the example above, where $\deg(c) = n$ for the center vertex $c$.

## 3 Violation of the Triangle Inequality

Bunke and Shearer [1998] have shown that

$$d(G, H) = 1 - \frac{|\text{MCS}(G, H)|}{\max\{|G|, |H|\}}, \quad (3)$$

where $|\text{MCS}(G, H)|$ is the size of a maximum common subgraph, is a metric and, in particular, fulfills the triangle inequality. The article by Schietgat et al. [2013] suggests that the size of a BBP-MCS combined with Eq. (3) is a metric, too. We show that this is not the case.

Consider the example shown in Figure 2 and let the size function of a graph $G$ be defined as $|G| = |V(G)| + |E(G)|$ following [Schietgat et al., 2013, Section 3.2, p. 364]. Employing BBP-MCS, we obtain $|\text{MCS}(G, H)| = 6$, $|\text{MCS}(H, F)| = 8$ and $|\text{MCS}(G, F)| = 1$ and accordingly:

| $d$ | $G$ | $H$ | $F$ |
|---|---|---|---|
| $G$ | 0 | 1/3 | 7/8 |
| $H$ | 1/3 | 0 | 1/9 |
| $F$ | 7/8 | 1/9 | 0 |

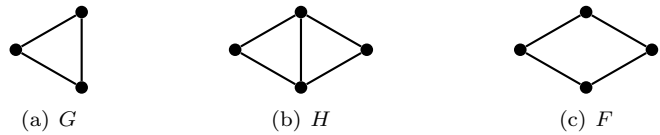

(a) $G$  (b) $H$  (c) $F$

Figure 2: Outerplanar graphs for which Eq. (3) does not satisfy the triangle inequality under BBP-MCS.

The triangle inequality is violated, since $d(G,F) > d(G,H) + d(H,F)$.


\end{document}